\documentclass[letterpaper, 10 pt, conference]{ieeeconf} 
\IEEEoverridecommandlockouts
\usepackage{amsmath,amssymb,amsfonts}
\usepackage{algorithmic}
\usepackage{graphicx}
\usepackage{bm}
\usepackage{textcomp}
\usepackage{xcolor}
\usepackage{xprintlen}
\usepackage{cases}
\usepackage{tabularx,ragged2e}
\usepackage{booktabs}
\usepackage{nccmath}
\usepackage{caption}
\usepackage{float}
\usepackage{subcaption}
\usepackage{mathtools}
\def\BibTeX{{\rm B\kern-.05em{\sc i\kern-.025em b}\kern-.08em
    T\kern-.1667em\lower.7ex\hbox{E}\kern-.125emX}}

\renewcommand{\vec}[1]{\bm{#1}}
\newcommand{\true}{\top}
\newcommand{\false}{\bot}

\newtheorem{problem}{\textbf{Problem}}
\newtheorem{proposition}{\bf{Proposition}}
\newtheorem{definition}{\bf{Definition}}
\newtheorem{lemma}{\bf{Lemma}}
\newtheorem{theorem}{\bf{Theorem}}
\newtheorem{remark}{\bf{Remark}}
\newtheorem{assumption}{\bf{Assumption}}

\newtheorem{example}{\textbf{Example}}
\newtheorem{Fact}{\textbf{Fact}}

\begin{document}

\title{Communication-Constrained STL Task Decomposition through\\ Convex Optimization}

\author{Gregorio Marchesini, Siyuan Liu, 
Lars Lindemann and Dimos V. Dimarogonas
	\thanks{This work was supported in part by the Horizon Europe EIC project SymAware (101070802), 
the ERC LEAFHOUND Project, the Swedish Research Council (VR), Digital Futures, and the Knut and Alice Wallenberg (KAW) Foundation.}
\thanks{Gregorio Marchesini, Siyuan Liu, and Dimos V. Dimarogonas are with the Division
of Decision and Control Systems, KTH Royal Institute of Technology, Stockholm, Sweden.
	E-mail: {\tt\small \{gremar,siyliu,dimos\}@kth.se}.  Lars Lindemann is with the Thomas Lord Department of Computer Science, University of Southern California, Los Angeles, CA, USA.
	E-mail: {\tt\small \{llindema\}@usc.ed}.   
}
}

\maketitle
\begin{abstract}
In this work, we propose a method to decompose signal temporal logic (STL) tasks for multi-agent systems subject to constraints imposed by the communication graph. Specifically, we propose to decompose tasks defined over multiple agents which require multi-hop communication, by a set of sub-tasks defined over the states of agents with 1-hop distance over the communication graph. To this end, we parameterize the predicates of the tasks to be decomposed as suitable hyper-rectangles. Then, we show that by solving a constrained convex optimization, optimal parameters maximising the volume of the predicate's super-level sets can be computed for the decomposed tasks. In addition, we provide a formal definition of conflicting conjunctions of tasks for the considered STL fragment and a formal procedure to exclude such conjunctions from the solution set of possible decompositions. The proposed approach is demonstrated through simulations.


\end{abstract}


\section{INTRODUCTION}
Temporal logics have recently received increasing attention as they allow for expressing high-level collaborative tasks among agents in a multi-agent system (MAS). Among the different temporal logics, signal temporal logic (STL) is an expressive language that has been successfully applied for both high-level planning  \cite{sun2022multi,barbosa2019guiding} and low-level feedback control \cite{lindemann2018control,LarsControl2} of multi-agent systems. The possibility of expressing both the temporal and spatial behavior of MASs without leveraging abstraction-based techniques makes STL particularly appealing for real-time control and planning by avoiding the curse of dimensionality.\par
At the current state of the art, three main low-level control approaches have been developed for the satisfaction of high-level STL tasks for MASs: prescribed performance control \cite{ppcLars,chen2023distributed}, mixed-integer linear programs (MILP) \cite{MILP1,MILP2,MILP3,MILP4}, and time-varying control barrier functions (CBF) \cite{lindemann2018control,LarsControl2}. 
In all the aforementioned works, it is taken for granted that the connectivity of the communication networks is ensured and always compatible with the task dependency graphs of MASs. 
Indeed, it is commonly assumed that each agent has either global state information about the system, or that agents involved in the same tasks can share state information by leveraging 1-hop communication. Nevertheless, this assumption is commonly violated in many real-world scenarios and fully decentralised control can not be achieved by the previous low-level control approaches when this assumption is not met. In this work, we propose a first step toward relaxing such assumption by decomposing tasks defined over multiple agents subject to multi-hop communication into conjunctions of sub-tasks that can be achieved by leveraging only 1-hop communication over the communication graph.\par

The results in \cite{charitidou2021signal} established a paradigm to decompose STL tasks defined over MASs into a set of tasks that can be independently satisfied by distinct sub-clusters of agents. However, the mismatch between the communication graph and task dependency graph of the MASs is not considered in the decomposition process, while each sub-cluster is considered to be fully connected. Similarly, \cite{leahy2022fast} and \cite{leahy2023rewrite} employ MILP to achieve the same type of decomposition as in \cite{charitidou2021signal}, but the underlying communication topology of the multi-agent system is again not considered, and thus decentralised control cannot be leveraged in the absence of fully connected sub-clusters. In \cite{MILP3}, communication constraints are considered in the motion planning approach of MASs. A MILP is solved to compute a valid state trajectory for the  MAS that satisfies a global STL task, while maintaining an optimized inter-agent communication quality of service. Nevertheless, perfect inter-agent communication is still assumed for agents under the same collaborative tasks.  

\par

The main contributions of this work are twofold. First, we propose a communication-constrained decomposition for a fragment of STL tasks defined over the absolute and relative state of the agents in the MAS. Namely, we exploit the communication graph of the MAS to decompose tasks whose predicate functions depend on the agents which require multi-hop communication into a conjunction of sub-tasks defined over the agents that have one 1-hop communication. Each newly introduced sub-task is defined over a parametric predicate function depending on the relative state of couples of agents in the system linked by the communication graph. 
By leveraging convex optimization tools, optimal parameters for each sub-task are then computed such that the volume of the super-level set of the predicate functions is maximised \cite{charitidou2021signal}. 
The decomposed tasks allow for a fully decentralised control approach for MASs. Second, we formally derive a set of conditions that can lead to un-satisfiable conjunctions of tasks in the considered STL fragments. We also provide a set of convex constraints to be included in the original convex optimization such that un-satisfiable conjunctions of tasks are excluded from the solution set of our task decomposition. \par

The paper is organized as follows. Section \ref{Preliminaries} presents preliminaries and Section \ref{problemformulation} introduces the problem definition. Section \ref{task decomposition} proposes our main task decomposition result. 
In Section \ref{conflicts}, we provide a formal definition of conflicting conjunctions of tasks and then propose a strategy to avoid such conflicts.  
Simulation results are provided in Section \ref{simulations} to demonstrate our decomposition approach, 
while concluding remarks are given in Section \ref{conclusions}.

\textit{Notation}: Bold letters indicate vectors while capital letters indicate matrices. Vectors are considered to be column vectors and the notation $\vec{x}[k]$ indicates the $k$-th element of $\vec{x}$. We define the minimum of a vector $\vec{a}\in \mathbb{R}^n$ element-wise as $\overset{\star}{\min}(\vec{a}) := \min_{k=1,\ldots n} \{\vec{a}[k]\}$. The notation $|\mathcal{A}|$ denotes the cardinality of the set $\mathcal{A}$, the symbol $\oplus$ indicates the Minkowski sum, 
the notation $\prod_{i=1}^k\mathcal{A}_i$ represents the Cartesian product of the sets $\mathcal{A}_i$ and $-\mathcal{A}=\{x|x=-v\; \forall v\in\mathcal{A}\}$. We denote the power set of $\mathcal{A}$ as $\mathcal{P}(\mathcal{A})$. The identity matrix of dimension $n$ is denoted as $I_n$. The set $\mathbb{R}_{+}$ denotes the non-negative real numbers.

\section{PRELIMINARIES}\label{Preliminaries}
Let $\mathcal{V}=\{1,\ldots N\}$ be the set of indices assigned to each agent in a  multi-agent system. We write the input-affine nonlinear dynamics for each agent  $i\in \mathcal{V}$ as
\begin{equation}\label{eq:agent dynamics}
\dot{\vec{x}}_i = \mathfrak{f}_i(\vec{x}_i) + \mathfrak{g}_i(\vec{x}_i)\vec{u}_i
\end{equation}
where $\vec{x}_i\in \mathbb{X}_i\subset\mathbb{R}^{n_i}$ is the state of the $i$-th agent and $\vec{u}_i\in \mathbb{U}_i\subset \mathbb{R}^{m_i}$ is the associated bounded control input. We assume, without loss of generality, that $n_i=n$, $\forall i\in\mathcal{V}$. Let $\mathfrak{f}_i : \mathbb{R}^{n} \rightarrow \mathbb{R}^{n}$ and $\mathfrak{g}_i : \mathbb{R}^{n} \rightarrow \mathbb{R}^{n \times m_i}$ be locally Lipschitz continuous functions of the agent state. We denote the full state of the system as $\vec{x}:=[\vec{x}_1^T,\vec{x}_2^T,\ldots \vec{x}_N^T]^T$. Given a control input $\vec{u}_i : [t_0,t_1]\rightarrow \mathbb{U}_i$, we define the state trajectory $\vec{x}_i: [t_0,t_1]\rightarrow \mathbb{X}_i$ for agent $i\in\mathcal{V}$ if $\vec{x}_i(t)$ satisfies \eqref{eq:agent dynamics} for all $t\in[t_0,t_1]$. 
We Also define the \textit{relative} state vector $\vec{e}_{ij}:=\vec{x}_j-\vec{x}_{i}$ for $i,j\in\mathcal{V}$.

\subsection{Signal Temporal Logic}
STL is a predicate logic applied to formally define spatial and temporal behaviours of real-valued continuous-time signals \cite{maler2004monitoring}. The atomic elements of STL are Boolean predicates $\mu:\mathbb{R}\rightarrow \{\true,\false\}$ defined as
$
\mu :=
\bigl\{\begin{smallmatrix*}
\true &\text{if} \; h(\vec{x})\geq 0 \\
\false &\text{if} \; h(\vec{x})<0 ,
\end{smallmatrix*}
$
where $h(\vec{x}): \mathbb{R}^n\rightarrow \mathbb{R}$ is a scalar valued \textit{predicate function}. The value of $h$ can generally depend upon the state of all the agents of the multi-agent system or a subset of the former. The STL grammar is recursively defined as:
$$
\phi::=\true|\mu|\lnot \phi|\phi_1 \land \phi_2|F_{[a,b]}\phi | G_{[a,b]}\phi| \phi_1 U_{[a,b]}\phi_2
$$
where $F_{[a,b]}$, $G_{[a,b]}$ and $U_{[a,b]}$ are the temporal \textit{eventually}, \textit{always} and \textit{until} operators over the time interval $[a,b]\subset \mathbb{R}_{+}$. We indicate that a state trajectory $\vec{x}(t)$ satisfies task $\phi$ at time $t$ as $\vec{x}(t)\models \phi$. The classical STL \textit{semantics} define the conditions such that $\vec{x}(t)\models\phi$ \cite{maler2004monitoring}. In the current work, we leverage the robust quantitative STL semantics: $\rho^{\mu}(\vec{x}, t)=h(\vec{x}(t)), \rho^{\neg \phi}(\vec{x}, t)=-\rho^\phi(\vec{x}, t), \rho^{\phi_1 \wedge \phi_2}(\vec{x}, t)= 
\min \left(\rho^{\phi_1}(\vec{x}, t), \rho^{\phi_2}(\vec{x}, t)\right),  \rho^{\phi_1 \mathcal{U}_{[a, b]} \phi_2}(\vec{x}, t)= 
\max _{t_1 \in[t+a, t+b]} \min \left(\rho^{\phi_2}\left(\vec{x}, t_1\right), \min _{t_2 \in\left[t, t_1\right]} \rho^{\phi_1}\left(\vec{x}, t_2\right)\right),\\  \rho^{F_{[a, b]} \phi}(\vec{x}, t)=\max _{t_1 \in[t+a, t+b]} \rho^\phi\left(\vec{x}, t_1\right), \rho^{G_{[a, b]} \phi}(\vec{x}, t)= \min _{t_1 \in[t+a, t+b]} \rho^\phi\left(\vec{x}, t_1\right),
$
and we recall the relations $\rho^{\phi}(\vec{x},0)>0\Rightarrow\vec{x}(t)\models \phi$ \cite{sadraddini2015robust,farahani2015robust,raman2015reactive}. In this work, we focus on STL tasks defined over predicate functions that depend on the state $\vec{x}_i$ of a single agent and/or on the relative state vector $\vec{e}_{ij}$ between two agents. Namely, we make use of the following STL fragment:
\begin{subequations}\label{eq:working fragment}
\begin{align}
\phi_{i} := F_{[a,b]}\mu_i|G_{[a,b]}\mu_i|\phi^{1}_{i} \land \phi^{2}_{i}, \label{eq:single agent spec}
 \\
\phi_{ij} := F_{[a,b]}\mu_{ij}|G_{[a,b]}\mu_{ij}|\phi^{1}_{ij}\land \phi^{2}_{ij}, \label{eq:multi agent spec}
\end{align}
\end{subequations}
where $
\mu_i :=
\bigl\{\begin{smallmatrix*}
\true &\text{if} \; h_i(\vec{x}_i)\geq 0 \\
\false &\text{if} \; h_i(\vec{x}_i)<0 ,
\end{smallmatrix*}
$, $\mu_{ij} :=
\bigl\{\begin{smallmatrix*}
\true &\text{if} \; h_{ij}(\vec{e}_{ij})\geq 0 \\
\false &\text{if} \; h_{ij}(\vec{e}_{ij})<0 ,
\end{smallmatrix*}
$ with predicate functions $h_{i}(\vec{x}_i)$ and $h_{ij}(\vec{e}_{ij})$. We refer to tasks as per \eqref{eq:single agent spec} as \textit{independent} and tasks as per \eqref{eq:multi agent spec} as \textit{collaborative}. Additionally, we identify the respective super-level sets as
\begin{equation}\label{eq:super level sets}
\begin{aligned}
        \mathcal{B}_i&=\{ \vec{x}_i\in\mathbb{X}_i|h_i(\vec{x}_i)\geq0\} ;\\
\mathcal{B}_{ij}&=\{ \vec{e}_{ij}\in\mathbb{X}_i\oplus(-\mathbb{X}_j)|h_{ij}(\vec{e}_{ij})\geq0\}.
\end{aligned}
\end{equation}
\subsection{Communication and Task Graphs}
We define an \textit{undirected} graph over the set of agents $\mathcal{V}$ as $\mathcal{G}(\mathcal{V},\mathcal{E})\in \Gamma$, where $\mathcal{E}\subseteq \mathcal{V}\times \mathcal{V}$ is the set of undirected edges of $\mathcal{G}$ and $\Gamma$ is the set of undirected graphs over the nodes $\mathcal{V}$. We define the set of neighbours of vertex $i$ as $\mathcal{N}(i)=\{j|(i,j) \in \mathcal{E} \land i\neq j\}$ and the extended neighbour set as $\bar{\mathcal{N}}(i):= \mathcal{N} \cup \{i\}$ such that self-loops are considered. Let the graph-valued function $\text{\textbf{add}}(\cdot,\cdot): \Gamma\times\mathcal{V}\times \mathcal{V}\rightarrow \Gamma$ be such that $\mathcal{G}'(\mathcal{E}',\mathcal{V})=\text{\textbf{add}}(\mathcal{G},\mathcal{Q})$ with $\mathcal{E}'=\mathcal{E}\cup\mathcal{Q}$ for a given set of edges $\mathcal{Q}\subset \mathcal{V}\times \mathcal{V}$. Similarly we define $\text{\textbf{del}}(\cdot,\cdot): \Gamma\times\mathcal{V}\times \mathcal{V}\rightarrow \Gamma$ such that $\mathcal{G}'(\mathcal{E}',\mathcal{V})=\text{\textbf{del}}(\mathcal{G},\mathcal{Q})$ and $\mathcal{E}'=\mathcal{E}\setminus\mathcal{Q}$. Let the vector $\vec{\pi}_i^j\in \mathcal{V}^l$ represent a \textit{directed} path of length $l$ defined as a vector of non-repeated indices in $\mathcal{V}$ such that $\vec{\pi}_i^j[k]\in \mathcal{V}\; \forall k=1,\ldots l$ , $\vec{\pi}_i^j[r] \neq \vec{\pi}_i^j[s] \;\forall s\neq r$ , $(\vec{\pi}_i^j[r],\vec{\pi}_i^j[r+1])\in \mathcal{E}$ and $(\vec{\pi}_i^j[1],\vec{\pi}_{i}^j[l])=(i,j)$. Let $\vec{\omega} \in \mathcal{V}^l$ represent a \textit{directed} closed cycle path, such that $\vec{\omega}$ is a \textit{directed} path and $\vec{\omega}[1]=\vec{\omega}[l]$. Let $\epsilon:\mathcal{V}^l\rightarrow \mathcal{P}(\mathcal{E})$ be a set-valued function such that $\epsilon(\vec{\pi}_i^j)=\{(\vec{\pi}_i^j[k],\vec{\pi}_i^j[k+1])| k=1,\ldots l-1\}$. We establish the relations
\begin{equation}\label{eq:edge sequence}
\vec{e}_{ij}=\sum_{(r,s)\in\epsilon(\bm{\pi}_{i}^j)}\bm{e}_{rs}, \quad \vec{0}=\sum\limits_{(r,s)\in\epsilon(\bm{\omega})}\bm{e}_{rs}.
\end{equation}
We further make a distinction between two types of graphs that refer to a multi-agent system. Namely, we define the \textit{communication graph} $\mathcal{G}_c(\mathcal{V},\mathcal{E}_c)\in \Gamma$ such that $(i,j)\in \mathcal{E}_c\subset \mathcal{V}\times \mathcal{V}$ if $i$ and $j$ are able to communicate their respective state to each other. It is assumed that each agent is always able to communicate with itself. We also define the \textit{task} graph $\mathcal{G}_\psi(\mathcal{V},\mathcal{E}_\psi)\in \Gamma$ such that $(i,j)\in \mathcal{E}_\psi$ if there exists a collaborative task $\phi_{ij}$ as per \eqref{eq:multi agent spec} between agent $i$ and $j$. Independent tasks $\phi_{i}$ as per \eqref{eq:single agent spec} induce self-loops in the task graph (see in Fig. \ref{fig:task and communication graph example}). We repeatedly make use of the subscript $c$ and $\psi$ to differentiate among properties of the communication and task graph, respectively. For instance, $\mathcal{N}_c(i)$ and $\mathcal{N}_\psi(i)$ indicate the neighbour set of agent $i$ in the communication graph and task graph, respectively. 
\begin{figure}[t]
    \centering
\includegraphics[width=0.45\textwidth]{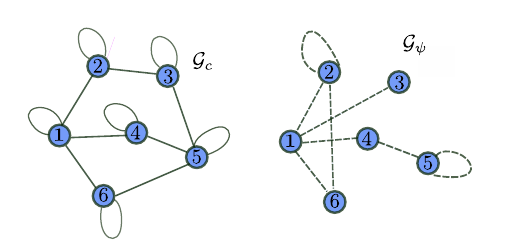}
    \caption{Simple example of communication (left) and task graph (right) for a multi-agent system with 6 agents. The task and communication graph are mismatching in their case.}
    \label{fig:task and communication graph example}
    \vspace{-0.5cm}
\end{figure}
We define $\psi$ as the \textit{global} task assigned to the MAS as
\begin{equation}\label{eq:global specification}
\psi := \bigwedge_{i=1}^N  \left(\phi_{i} \land \bigwedge_{j\in \mathcal{N}_\psi(i)}\phi_{ij}\right).
\end{equation}
Global tasks according to \eqref{eq:global specification} are particularly suitable for defining, e.g. time-varying relative formations. We propose an example to clarify the notation introduced so far.
\begin{example}
Consider the communication and task graphs in Fig \ref{fig:task and communication graph example}. Agents $2$ and $5$ have an independent task $\phi_{2}$, $\phi_{5}$ respectively, while $\phi_{12}, \phi_{13}, \phi_{14}, \phi_{16}, \phi_{56}, \phi_{26}$ are collaborative tasks. The path $\vec{\pi}_1^3=[1,2,3]$ connects agent $1$ to 3 with length $l=3$ and $\epsilon(\vec{\pi}_1^3)=\{(1,2),(2,3)\}$. Likewise, $\vec{\pi}_{6}^2=[6,1,2]$ and $\epsilon(\vec{\pi}_6^2)=\{(6,1),(1,2)\}$. Agents 1, 2, 6 form a cycle of tasks $\vec{\omega}=[1,2,6,1]$ in $\mathcal{G}_{\psi}$.
\end{example}
We conclude this section by stating the three main assumptions that hold throughout the work:
\begin{assumption}(Connectivity)\label{connettivity assumption}
    The communication graph $\mathcal{G}_c$ is a time-invariant connected graph.
\end{assumption}
\begin{assumption}(Concavity)\label{concavity}
    The predicate functions $h_{ij}$ and $h_i$ are concave functions of $\vec{e}_{ij}$ and $\vec{x}_i$ respectively.
\end{assumption}
\begin{assumption}(Task symmetry)\label{task symmetry}
For each STL task $\phi_{ij}$ in \eqref{eq:multi agent spec}, we have that $\phi_{ij}=\phi_{ji}\; \forall (i,j)\in\mathcal{E}_\psi$.
\end{assumption}
The first assumption is required to obtain a decomposition and ensure that such decomposition remains valid over time. We leave the case of time-varying $\mathcal{G}_c$ as future work. The second assumption is required to maintain the super-level sets $\mathcal{B}_{ij}$ in \eqref{eq:super level sets} convex. 
Note that the same assumptions were considered in \cite{charitidou2021signal} for a similar task decomposition. The last Assumption \ref{task symmetry} is not restrictive since the $\mathcal{G}_\psi$ is undirected.

\section{PROBLEM FORMULATION}\label{problemformulation}
Previous works in STL control often assume that $\mathcal{G}_c$ is fully connected or that $\psi$ can be decomposed into a conjunction of tasks defined over the state of fully connected sub-clusters of agents in $\mathcal{G}_c$. In the current work, we seek to drop these assumptions. Indeed, given the structure of $\psi$ as a conjunction of collaborative tasks among couples of agents or independent tasks, then the satisfaction of $\psi$ can be achieved by designing a fully decentralised low-level controller for each agent $i$, requiring only state information from the 1-hop neighbours in $\mathcal{N}_c(i)$ when $\mathcal{N}_\psi(i)\subseteq\mathcal{N}_c(i)$ \cite{LarsControl2}. However, a decentralised control approach that can guarantee the satisfaction of $\psi$ when $\mathcal{N}_\psi(i)\not\subseteq\mathcal{N}_c(i)$ is not available in the literature. Thus, we wish to exploit $\mathcal{G}_c$ to construct a new task $\bar{\psi}$ in the form of \eqref{eq:global specification} and a new graph $\mathcal{G}_{\bar{\psi}}$ such that $\mathcal{N}_{\bar{\psi}}(i)\subseteq\mathcal{N}_c(i)$ and $(\vec{x},t)\models\bar{\psi}\Rightarrow (\vec{x},t)\models\psi$. With this goal, we decompose the original tasks $\phi_{ij}$, with $(i,j)\in\mathcal{E}_\psi\setminus\mathcal{E}_c$, over the relative states of a path of agents $\vec{\pi}_i^j$ over $\mathcal{G}_c$. We formalise the problem as follows: 

\begin{problem}\label{problem1}
Consider the multi-agent system with agents in $\mathcal{V}$, communication graph $\mathcal{G}_c$ and task graph $\mathcal{G}_{\psi}$ such that $\psi$ is according to \eqref{eq:global specification} and $\mathcal{E}_\psi\setminus\mathcal{E}_c\neq \emptyset$. Then, define a new global task $\bar{\psi}$ in the form
\begin{equation}\label{eq:global specification rewritten}
\bar{\psi} := \bigwedge_{i=1}^N  \left(\phi_{i} \land \bigwedge_{j\in \mathcal{N}_\psi(i)\cap \mathcal{N}_c(i)}\phi_{ij} \wedge \bigwedge_{j\in \mathcal{N}_\psi(i)\setminus\mathcal{N}_c(i)} \phi^{\pi_i^j}\right),
\end{equation}
with the new task graph being $\mathcal{G}_{\bar{\psi}}$, such that $\forall (i,j) \in \mathcal{E}_\psi\setminus \mathcal{E}_c$ 
\begin{equation}\label{eq:path specification}
\phi^{\pi_i^j}:=\bigwedge_{(r,s)\in\epsilon(\pi_i^j)} \bar{\phi}^{\pi_i^j}_{rs},
\end{equation}
 where $\vec{\pi}_{i}^j$ are paths in $\mathcal{G}_c$ such that $\epsilon(\pi_{i}^j)\subset \mathcal{E}_c$ and $\bar{\phi}^{\pi_i^j}_{rs}$ are tasks of type \eqref{eq:multi agent spec} to be appropriately defined such that $\vec{x}(t) \models \bar{\psi} $ implies $ \vec{x}(t)\models \psi$.
\end{problem}
\begin{example}
    Consider Fig \ref{fig:path specification image}. Agents 1 and 4 share a collaborative task $\phi_{14}$, which cannot be directly satisfied as the two agents are not communicating (indicated by a dashed line). The task $\phi_{14}$ is then replaced by the task $\phi^{\pi_1^4}= \bar{\phi}^{\pi_1^4}_{12}  \land \bar{\phi}^{\pi_1^4}_{23} \land \bar{\phi}^{\pi_1^4}_{34}$ where $\pi_{1}^4=[1,2,3,4]$.
\end{example}
\begin{figure}
    \centering
    \includegraphics[scale=0.66,trim={0 1.3cm 0 1.3cm},clip]{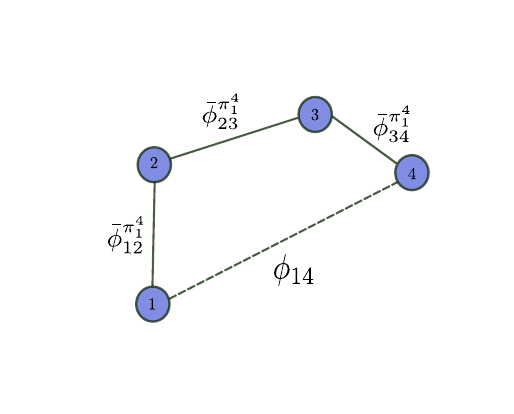}
    \caption{Example of decomposition according to \eqref{eq:path specification}.}
    \label{fig:path specification image}
    \vspace{-0.4cm}
\end{figure}
In the next sections, we will develop on how the tasks $\bar{\phi}^{\pi_i^j}_{rs}$ are defined according to $\mathcal{G}_c$ and the temporal properties of the collaborative tasks $\phi_{ij}$ with $(i,j)\in\mathcal{E}_{\psi}\setminus\mathcal{E}_{c}$.


\section{TASK DECOMPOSITION}\label{task decomposition}
We outline the task decomposition approach applied to obtain a single task $\phi^{\pi_{i}^{j}}$, according to $\eqref{eq:path specification}$, from a given collaborative task $\phi_{ij}=T_{[a,b]}\mu_{ij}$ with $T\in\{G,F\}$ as defined in \eqref{eq:multi agent spec} (see Fig. \ref{fig:path specification image}). In the case where $\phi_{ij}$ contains conjunctions of tasks as $\phi_{ij}=\bigwedge_k\phi_{ij}^{k}$, the method developed in this section is applied for each $\phi_{ij}^k$, leading to $\phi_{ij}$ being decomposed as $\phi^{\pi_i^j} = \bigwedge_k(\phi^{\pi_i^j})^{k} $. Each $(\phi^{\pi_i^j})^{k}$ is obtained with the same approach applied in the single task case presented in this section. We further elaborate on this point in Remark \ref{remark multiple formulas}.
\subsection{Path decomposition of STL tasks}
Constructing task $\phi^{\pi_i^j}$ requires two steps: 1) find a path $\vec{\pi}_{i}^j$ from agents $i$ to $j$ through $\mathcal{G}_c$, 2) find a family of tasks $\bar{\phi}^{\pi_{i}^{j}}_{rs}$ that can be applied to construct $\phi^{\pi_{i}^{j}}$ as per \eqref{eq:path specification} such that $\vec{x}(t)\models \phi^{\pi_{i}^{j}}\Rightarrow \vec{x}(t)\models \phi_{ij}$. The solution to step 1 is readily available from the literature as there exists a plethora of algorithms for finding a path from node $i$ to node $j$ in a connected graph \cite{lavalle2006planning}. We highlight that finding the shortest path connecting two nodes is not a requirement for our work and thus we select the well-known Dijkstra algorithm. On the other hand, for the solution of step 2, we leverage axis-aligned $n$-dimensional hyper-rectangles as suitable predicate functions for the tasks $\bar{\phi}^{\pi_{i}^{j}}_{rs}$ \cite{ziegler2012lectures,froitzheim2016efficient}. The following definitions and properties of hyper-rectangles are provided:
\begin{definition}(\cite[Ch .1.1]{ziegler2012lectures}\cite[Def 3.6]{froitzheim2016efficient})\label{hyper-rectangle}
 Given $\vec{\nu}\in\mathbb{R}^n$ such that $\vec{\nu}[k]\in \mathbb{R}_+ \quad \forall k\in 1,\ldots n$ and $\vec{p}\in \mathbb{R}^n$, an axis-aligned hyper-rectangle $\mathcal{H}(\vec{p},\vec{\nu})$ is defined as $\mathcal{H}(\vec{p},\vec{\nu})=\prod_{k=1}^n[\vec{p}[k]-\vec{\nu}[k]/2,\vec{p}[k]+\vec{\nu}[k]/2]$. Equivalently,  $\mathcal{H}(\vec{p},\vec{\nu})=\{\vec{\zeta}\in\mathbb{R}^n|A(\vec{\zeta}-\vec{p})-\vec{b}(\vec{\nu})\geq\vec{0}\}$ such that $\vec{b}(\vec{\nu}) =[\vec{\nu}^T/2,-\vec{\nu}^T/2]$ and $A = [I_n,-I_n]^T$.
\end{definition}
\begin{proposition}(\cite[Ch. 1.1]{ziegler2012lectures})\label{convex combination}
Any point $\vec{\zeta}\in\mathcal{H}(\vec{p},\vec{\nu})$ is a convex combination of the set of vertices $\mathcal{P}(\vec{p},\vec{\nu}):=\{\vec{v}\in \mathbb{R}^n | \vec{v}[s]=\vec{p}[s] + \vec{\nu}[s]/2 \; \text{or} \; \vec{v}[s]=\vec{p}[s] - \vec{\nu}[s]/2 \; \forall s=1,\ldots n \}$, where $|\mathcal{P}(\vec{p},\vec{\nu})|=2^n$, such that $\vec{\zeta}=\sum_{i=1}^{2^n}\lambda_i\vec{v}_i$, with $\sum_{i}^{2n}\lambda_i=1$, $0\geq\lambda_i\geq1$ and $\vec{v}_i\in \mathcal{P}(\vec{p},\vec{\nu})$.
\end{proposition}
\begin{proposition}(\cite{froitzheim2016efficient})\label{hyper mink}
Let two axis aligned hyper-rectangles $\mathcal{H}^1(\vec{p}_1,\vec{\nu}_1)$, $\mathcal{H}^2(\vec{p}_2,\vec{\nu}_2)$, then the Minkowski sum 
$\mathcal{H}^3(\vec{p}_3,\vec{\nu}_3)=\mathcal{H}^1(\vec{p}_1,\vec{\nu}_1)\oplus\mathcal{H}^2(\vec{p}_2,\vec{\nu}_2)$ is an axis aligned hyper-rectangle such that $\vec{p}_3 = \vec{p}_1+\vec{p}_2$ and $\vec{\nu}_3 = \vec{\nu}_1+\vec{\nu}_2$.
\end{proposition}
\begin{proposition}(\cite{ziegler2012lectures})\label{low inclusion}
Consider a concave scalar-valued function $g:\mathbb{R}^n\rightarrow \mathbb{R}$ and hyper-rectangle $\mathcal{H}(\vec{p},\vec{\nu})$. If $g(\vec{v}_i)\geq0\; \forall \vec{v}_i\in \mathcal{P}(\vec{p},\vec{\nu})$ then $g(\vec{\zeta})\geq0 \;\forall \vec{\zeta}\in\mathcal{H}(\vec{p},\vec{\nu})$.
\end{proposition}
If for each task $\bar{\phi}^{\pi_{i}^{j}}_{rs}$ with $(r,s)\in\epsilon(\vec{\pi}_i^j)$ we define a centre $\vec{p}^{\pi_i^j}_{rs}$ and a dimension vector $\vec{\nu}^{\pi_i^j}_{rs}$ as parameters, we can employ the following family of concave predicate functions with respective predicate and super-level set:
\begin{subequations}\label{eq:infty predicate function}
\begin{align}
&\bar{h}^{\pi_i^j}_{rs}(\vec{e}_{rs},\vec{\eta}^{\pi_i^j}_{rs}) :=\overset{\star}{\min}(A(\vec{e}_{rs}-\vec{p}^{\pi_i^j}_{rs})-\vec{b}(\vec{\nu}^{\pi_i^j}_{rs}))\geq0, \label{eq:parameteric predicate function} \\
&\bar{\mathcal{B}}^{\pi_i^j}_{rs}(\vec{\eta}^{\pi_i^j}_{rs}):=\{\vec{e}_{rs}\in \mathbb{X}_r\oplus (-\mathbb{X}_s)|\bar{h}^{\pi_i^j}_{rs}(\vec{e}_{rs},\vec{\eta}^{\pi_i^j}_{rs})\geq0\}, \label{eq:parameteric superlevel set}\\
&\bar{\mu}^{\pi_i^j}_{rs}(\vec{\eta}^{\pi_i^j}_{rs}) := \begin{cases}
\true& \; \text{if} \quad \bar{h}^{\pi_i^j}_{rs}(\vec{e}_{rs},\vec{\eta}^{\pi_i^j}_{rs})\geq 0 \\
\false& \; \text{if} \quad \bar{h}^{\pi_i^j}_{rs}(\vec{e}_{rs},\vec{\eta}^{\pi_i^j}_{rs})< 0,
\end{cases}\label{eq:parameteric predicate}
\end{align}
\end{subequations}
where $\vec{\eta}_{rs}^{\pi_i^j}:=[(\vec{p}_{rs}^{\pi_i^j })^T,(\vec{\nu}_{rs}^{\pi_i^j})^T]^T$ are free parameters. We recall that $\overset{\star}{\min}$ is the element-wise minimum and $\overset{\star}{\min}(A(\vec{\zeta}-\vec{p})-\vec{b}) \geq0\Rightarrow A(\vec{\zeta}-\vec{p})-\vec{b}\geq \vec{0}$ for a given $\vec{\zeta}$. The vector $\vec{\eta}_{rs}^{\pi_i^j}$ is computed as a result of the convex program defined in Theorem \ref{convex optimization theorem}. The parametric set $\bar{\mathcal{B}}^{\pi_i^j}_{rs}$ in \eqref{eq:parameteric superlevel set} is a convex hyper-rectangle by Def. \ref{hyper-rectangle} with volume given by $\prod_{s=1}^n\vec{\nu}^{\pi_i^j}_{rs}[s]$ \cite{farahani2015robust}.
Hyper-rectangles are particularly suitable for the task decomposition due to their efficient vertex representation and Minkowski computation in Proposition \ref{hyper mink} \cite[Sec. 3.6]{farahani2015robust}.
In principle, any zonotope can be used for the decomposition presented as we clarify in Remark \ref{remark multiple formulas}. We can now state our first decomposition result:
\begin{lemma}\label{single formula decomposition lemma}
Consider a task $\phi_{ij}=T_{[a,b]}\mu_{ij}$ with $T\in\{G,F\}$ according to $\eqref{eq:multi agent spec}$, the corresponding predicate function $h_{ij}$ satisfying Assumption \ref{concavity} and $\mathcal{B}_{ij}$ according to \eqref{eq:super level sets}. Further consider a path $\vec{\pi}_i^j$ through the communication graph $\mathcal{G}_c$ and $\phi^{\pi_i^j}=\bigwedge\nolimits_{(r,s)\in \epsilon(\vec{\pi}_{i}^j)}\bar{\phi}_{rs}^{\pi_i^j}$ such that each $\bar{\phi}_{rs}^{\pi_i^j}$ is defined as
\begin{subnumcases}{\label{eq:specific form of phi bar}\bar{\phi}^{\pi_{i}^{j}}_{rs} := }
F_{[a^*,b^*]}\bar{\mu}^{\pi_{i}^{j}}_{rs} \; \text{if} \;\; T=F\label{eq:eventually case}\\
G_{[a^*,b^*]}\bar{\mu}^{\pi_{i}^{j}}_{rs}\; \text{if} \;\;  T=G,\label{eq:always case},
\end{subnumcases}
with $[a^*,b^*]$ defined as
\begin{subnumcases}{\label{eq:time intervals definition}[a^*,b^*] :=}
 \bigl[\bar{t} , \bar{t}\bigr]  \text{ with $\bar{t}\in[a,b]$ if} \; T=F \label{eq:eventually time interval}\\[1ex]
 \bigl[a , b\bigr] \text{ if} \; T=G, \label{eq:always time interval}
\end{subnumcases}
where the predicate $\bar{\mu}^{\pi_{i}^{j}}_{rs}(\vec{\eta}^{\pi_i^j}_{rs})$, predicate function $\bar{h}^{\pi_{i}^{j}}_{rs}(\vec{e}_{rs},\vec{\eta}^{\pi_i^j}_{rs})$ and super-level set $\mathcal{\bar{B}}^{\pi_{i}^{j}}_{rs}(\vec{\eta}^{\pi_i^j}_{rs})$ are as per \eqref{eq:infty predicate function}.
If $\bar{\phi}^{\pi_i^j}_{rs}$ are defined according to \eqref{eq:specific form of phi bar}-\eqref{eq:time intervals definition} and
 \begin{equation}\label{eq:minkosky inclusion}
 \bigoplus_{(r,s)\in\epsilon(\vec{\pi}_i^j)} \bar{\mathcal{B}}^{\pi_{i}^{j}}_{rs}(\vec{\eta}^{\pi_i^j}_{rs})\subseteq \mathcal{B}_{ij};
 \end{equation}
then $\vec{x}(t)\models \phi^{\pi_i^j}\Rightarrow\vec{x}(t)\models \phi_{ij}$.
\end{lemma}
\begin{proof}
We prove the lemma for  $\phi_{ij}:=F_{[a,b]}\mu_{ij}$ while the case of $\phi_{ij}:=G_{[a,b]}\mu_{ij}$ follows a similar procedure. We omit the dependency of $\bar{h}_{rs}^{\pi_i^j},\mathcal{B}_{rs}^{\pi_i^j}$ from $\vec{\eta}_{rs}^{\pi_i^j}$ to reduce the burden of notation. Given the path $\vec{\pi}_i^j$ over $\mathcal{G}_c$ we define $\phi^{\pi_i^j}=\bigwedge\nolimits_{(r,s)\in \epsilon(\vec{\pi}_{i}^j)}\bar{\phi}_{rs}^{\pi_i^j}=\bigwedge_{(r,s)\in\epsilon(\vec{\pi}_i^j)}F_{[\bar{t},\bar{t}]}\bar{\mu}^{\pi_i^j}_{rs}$ according to \eqref{eq:specific form of phi bar}-\eqref{eq:time intervals definition}, where $[\bar{t},\bar{t}]\subseteq[a,b]$. It is known that $\vec{x}(t)\models\phi^{\pi_i^j} \Rightarrow \rho^{\phi^{\pi_i^j}}(\vec{x},0)=\min_{(r,s)\in\epsilon(\vec{\pi}_i^j)}\{\rho^{\bar{\phi}_{rs}^{\pi_i^j}}(\vec{x},0)\}>0$. By definition of robust semantics for the $F$ operator we know that $\rho^{\bar{\phi}_{rs}^{\pi_i^j}}(\vec{x},0)>0 \Rightarrow \exists t_{rs}\in[\bar{t},\bar{t}] \, \text{such that}\; \bar{h}^{\pi_i^j}_{rs}(\vec{e}_{rs}(t_{rs}))> 0 \;\forall (r,s)\in \epsilon(\pi_i^j)$. Since the interval $[\bar{t},\bar{t}]$ only contains $\bar{t}$, then $t_{rs} = \bar{t}\; \forall (r,s)\in \epsilon(\vec{\pi}_i^j)$. Hence $\rho^{\bar{\phi}_{rs}^{\pi_i^j}}(\vec{x},0)>0\Rightarrow \bar{h}^{\pi_i^j}_{rs}(\vec{e}_{rs}(\bar{t}))> 0 \Rightarrow \vec{e}_{rs}(\bar{t})\underset{\text{\eqref{eq:parameteric superlevel set}}}{\in}\mathcal{\bar{B}}^{\pi_i^j}_{rs}\;\forall (r,s)\in \epsilon(\pi_i^j)$. From \eqref{eq:edge sequence} and \eqref{eq:minkosky inclusion} we have $\vec{e}_{ij}(\bar{t})\underset{\text{\eqref{eq:edge sequence}}}{=}\sum_{(r,s)\in \epsilon(\vec{\pi}_i^j)}\vec{e}_{rs}(\bar{t})\in \bigoplus_{(r,s)\in \epsilon(\vec{\pi}_i^j)} \bar{\mathcal{B}}^{\pi_{i}^{j}}_{rs} \underset{\text{\eqref{eq:minkosky inclusion}}}{\subseteq} \mathcal{B}_{ij}\; \Rightarrow \vec{e}_{ij}(\bar{t})\in\mathcal{B}_{ij} ]\Rightarrow h_{ij}(\vec{e}_{ij}(\bar{t}))>0$. We thus arrived at the conclusion since $\bar{t}\in [a,b]$ and the robust semantics for the $F$ operator we have $\rho^{\phi_{ij}}(\vec{x},0)=\max_{t\in[a,b]} (h_{ij}(\vec{e}_{ij}(t)))>0 \Rightarrow \vec{x}(t)\models \phi_{ij}$.
\end{proof}
We highlight that $\bigoplus_{(r,s)\in \epsilon(\vec{\pi}_i^j)} \bar{\mathcal{B}}^{\pi_{i}^{j}}_{rs}$ is an axis-aligned hyper-rectangles according to Prop. \ref{hyper mink}. Moreover, from Assumption \ref{concavity}, each $h_{ij}$ is concave, such that satisfying \eqref{eq:minkosky inclusion} consists in verifying the $2^n$ convex relations $-h_{ij}(\vec{v})\leq0$ over the vertices $\vec{v}\in\mathcal{P}(\vec{p}^{\pi_i^j},\vec{\nu}^{\pi_i^j})$, where $\vec{p}^{\pi_i^j}=\sum_{(r,s)\in\epsilon(\vec{\pi}_i^j)} \vec{p}^{\pi_i^j}_{rs}$ and $\vec{\nu}^{\pi_i^j}=\sum_{(r,s)\in\epsilon(\vec{\pi}_i^j)} \vec{\nu}^{\pi_i^j}_{rs}$ as per Prop. \ref{hyper mink}.

\subsection{Computing optimal parameters}
In Lemma \ref{single formula decomposition lemma} we showed that \eqref{eq:specific form of phi bar}, \eqref{eq:time intervals definition} and \eqref{eq:minkosky inclusion} imply $\vec{x}(t)\models \phi^{\pi_i^j}\Rightarrow \vec{x}(t)\models \phi_{ij}$. In Thm. \ref{convex optimization theorem}, we present a procedure to compute the parameters $\vec{\eta}^{\pi_i^j}_{rs}$ for each task $\bar{\phi}^{\pi_i^j}_{rs}$ such that the \eqref{eq:specific form of phi bar}, \eqref{eq:time intervals definition} and \eqref{eq:minkosky inclusion} are satisfied and the volumes of the super-level set $\bar{\mathcal{B}}^{\pi_i^j}_{rs}$ are maximized. Namely, given a single hyper-rectangle $\bar{\mathcal{B}}^{\pi_i^j}_{rs}$, it is possible to maximize its volume $\prod_{k=1}^n\vec{\nu}_{rs}^{\pi_i^j}[k]$, by minimizing  $(\prod_{k=1}^n\vec{\nu}_{rs}^{\pi_i^j}[k])^{-1}$, which is a convex function since $\vec{\nu}_{rs}^{\pi_i^j}[k]>0 \; \forall k=1,\ldots n$. For convenience, we introduce the set $\Theta^{\pi_i^j}:=\{\vec{\eta}^{\pi_i^j}_{rs} | \forall (r,s)\in \epsilon(\pi_{i}^j)\}$ which gathers all the parameter along a path $\vec{\pi}_i^j$ applied for the decomposition of a given task $\phi_{ij}$ and $\Theta=\bigcup_{(i,j)\in \mathcal{E}_{\psi}\setminus \mathcal{E}_c}\Theta^{\pi_i^j}$ is then the set that of all the parameters applied for the decomposition. We are now ready to present our second result.

\begin{theorem}\label{convex optimization theorem}
    Consider a multi-agent system with agents in $\mathcal{V}$ and subject to a global task $\psi$ according to \eqref{eq:working fragment}. Further consider the associated task and communication graphs $\mathcal{G}_c$, $\mathcal{G}_\psi$ such that $\mathcal{E}_\psi \setminus\mathcal{E}_c\neq \emptyset$ and $\mathcal{G}_c$ respects Assumption \ref{connettivity assumption}. For all $(i,j)\in \mathcal{E}_\psi \setminus\mathcal{E}_c$ consider the paths $\vec{\pi}_i^j$  in $\mathcal{G}_c$ and the corresponding task $\phi^{\pi_i^j}$ satisfying the conditions in Lemma \ref{convex optimization theorem}. Define the following convex optimization problem 
    \begin{subequations}\label{eq:convex optimization problem}
    \begin{align}
    \min_{\vec{\eta}_{rs}^{\pi_i^j}\in\Theta} \;  \sum_{(i,j)\in \mathcal{E}_{\psi}\setminus\mathcal{E}_{c}}\sum_{(r,s)\in \epsilon(\pi_i^j)} \Bigl(\prod_{k=1}^{n}\vec{\nu}^{\pi_i^j}_{rs}[k]\Bigr)^{-1}
   \\
   \bigoplus_{(r,s)\in \epsilon(\vec{\pi}_{i}^j)} \bar{\mathcal{B}}^{\pi_{i}^{j}}_{rs}(\vec{\eta}^{\pi_{i}^{j}}_{rs})\subseteq \mathcal{B}_{ij} \; \forall (i,j)\in \mathcal{E}_\psi\setminus \mathcal{E}_c,  \label{eq:super level set inclusion constraint}
    \end{align}
    \end{subequations}
Assuming feasibility of \eqref{eq:convex optimization problem} and that there exists $\vec{x}(t)$ such that $\vec{x}(t)\models \bar{\psi}$ with $\bar{\psi}$ according to \eqref{eq:global specification rewritten}, then $\vec{x}(t)\models \psi$.
 \end{theorem}
\begin{proof}
Given that a solution to \eqref{eq:convex optimization problem} exists, then the satisfaction of the constraints set \eqref{eq:super level set inclusion constraint} implies that condition \eqref{eq:minkosky inclusion} is satisfied for all $(i,j)\in\mathcal{E}_{\psi}\setminus\mathcal{E}_{c}$. The conditions from Lemma 1 are then satisfied for all $\phi^{\pi_i^j}$ with $(i,j)\in\mathcal{E}_{\psi}\setminus\mathcal{E}_{c}$. We now analyse the satisfaction of $\psi$ and $\bar{\psi}$ through the definition of the robust semantics such that
$
\rho^{\psi}= \min \{\min\limits_{i\in\mathcal{V}}\{\rho^{\phi_i}\},\min_{(i,j)\in\mathcal{E}_\psi}\{\rho^{\phi_{ij}}\}\}
$ and
$
\rho^{\bar{\psi}} = \min \{\min\limits_{i\in\mathcal{V}}\{\rho^{\phi_i}\},\min_{(i,j)\in\mathcal{E}_\psi\cap\mathcal{E}_c}\{\rho^{\phi_{ij}}\}\min_{(i,j)\in\mathcal{E}_\psi\setminus\mathcal{E}_c}\{\rho^{\phi^{\pi_i^j}}\}\};$ where we have omitted the dependency from $(\vec{x},0)$. 
Assuming that $\vec{x}(t)\models\bar{\psi}$, then $\rho^{\phi_{ij}}(\vec{x},0)>0 \, \forall (i,j)\in \mathcal{E}_{\psi}\cap\mathcal{E}_c$ and $\rho^{\phi_{i}}(\vec{x},0)>0$. Furthermore, we know from Lemma 1 that $\rho^{\phi^{\pi_i^j}}(\bm{x},0)>0\Rightarrow\rho^{\phi_{ij}}(\bm{x},0)>0\; \forall (i,j)\in \mathcal{E}_{\psi}\setminus\mathcal{E}_{c}$ and eventually $\rho^{\phi_{ij}}(\vec{x},0)>0 \; \forall (i,j)\in \mathcal{E}_\psi \Rightarrow \vec{x}(t)\models \bar{\psi}\Rightarrow \vec{x}(t)\models \psi$. 
\end{proof}
If we define the set of edges involved in the decomposition of $\psi$ as $\mathcal{Q}:=\bigcup_{(i,j)\in\mathcal{E}_{\psi}\setminus\mathcal{E}_c}\epsilon(\vec{\pi}_i^j)$ then we can write $\mathcal{G}_{\bar{\psi}}$ as a function of $\mathcal{G}_{\psi}$ as
$\mathcal{G}_{\bar{\psi}} = \text{\textbf{add}}\bigl(\text{\textbf{del}}(\mathcal{G}_\psi,\mathcal{E}_\psi\setminus\mathcal{E}_c),\mathcal{Q}\bigr),$
which correspond to deleting all the edges in $\mathcal{G}_\psi$ not corresponding to an edge in $\mathcal{G}_c$, while we add all the edges from $\mathcal{G}_c$ that are introduced by the paths $\vec{\pi}_i^j$ during the decomposition. We have thus deduced a procedure that solves Problem \ref{problem1}.
\begin{remark}\label{remark multiple formulas}
    Problem \eqref{eq:convex optimization problem} handles cases in which $\phi_{ij}$ has conjunctions. Indeed, if $\phi_{ij}=\land_{k=1}^p\phi_{ij}^k$ for some $p\geq1$ then we define a task $(\phi^{\pi_i^j})^k$ as per Lemma \ref{single formula decomposition lemma} and a constraint as per \eqref{eq:super level set inclusion constraint} for each $k=1,\ldots p$ to be introduced in \eqref{eq:convex optimization problem}.
\end{remark}
\begin{remark}\label{remark multiple formulas}
    In principle, any type of zonotope can be employed for the decomposition approach developed. However, we remark that the cost of imposing constraint \eqref{eq:super level set inclusion constraint} increases with the number of vertices defining the Minkowski sum in the left-hand side of  \eqref{eq:super level set inclusion constraint}. 
\end{remark}
We outline that although \eqref{eq:convex optimization problem} might yield a solution, there is no guarantee that the tasks $\phi^{\pi_i^j}$ computed from \eqref{eq:convex optimization problem} in Thm. \ref{convex optimization theorem} can be satisfied in conjunction with each other or that they can be satisfied in conjunction with the un-decomposed formulas $\phi_{ij}\;\forall (i,j)\in \mathcal{E}_{\psi}\cap\mathcal{E}_c$. In other words, the new task $\bar{\psi}$ obtained from the optimization problem presented in Thm. \ref{convex optimization theorem} is not guaranteed to be satisfiable. We analyse this problem in the next section.

\section{Dealing with conflicts}\label{conflicts}
\subsection{Conflicting conjuncitons}
We consider the following notion of conflicting conjunction for formulas defined by the STL fragment \eqref{eq:working fragment}:
\begin{definition}(Conflicting conjunction)\label{conflicting conjunction}
A conjunction of formulas $\bigwedge_{k}\phi^{k}_{ij}$, where $\phi^{k}_{ij}$ is according to  \eqref{eq:multi agent spec}, is a \textit{conflicting conjunction} if there does not exist a state trajectory $\vec{x}(t)$ for the multi-agent system such that $\vec{x}(t)\models \bigwedge_{k}\phi^k_{ij}$.
\end{definition}
In this section, we first state the four types of conflicting conjunctions that we consider for fragment in \eqref{eq:working fragment} and we conjecture that these are the only 4 possible types. Second, we provide sufficient conditions that can be enforced as convex constraints in Theorem \ref{convex optimization theorem} such that the obtained $\bar{\psi}$ does not suffer from conflicting conjunctions. We assume that for the original global task $\psi$, none of these conflicts arise, while they could arise due to the introduction of the new formulas $\phi^{\pi_i^j}$ during the decomposition process outlined in Thm. \ref{convex optimization theorem}. For ease of notation, in this section, we drop the notation $\bar{\phi}^{\pi_i^j}_{rs}$ and we rewrite all the tasks $\phi_{ij}$ in $\mathcal{G}_{\bar{\psi}}$ for a single edge $(i,j)\in \mathcal{E}_{\bar{\psi}}$ as
\begin{subequations}
\begin{align}\label{eq: splitted specification}
\phi_{ij}=\bigwedge_{k\in \mathcal{I}^g_{ij}}\phi^k_{ij} \land \bigwedge_{k\in \mathcal{I}^f_{ij}}\phi^k_{ij}
\\ \label{eq: always splitted specification}
\phi^k_{ij}=G_{[a^k,b^k]}\mu^k_{ij} \;\text{if $k\in  \mathcal{I}^g_{ij}$}, 
\\ \label{eq: eventually splitted specification}
\phi^k_{ij}=F_{[a^k,b^k]}\mu^k_{ij} \quad \text{if $k\in  \mathcal{I}^f_{ij}$}.
\end{align}
\end{subequations}
such that $\mathcal{I}_{ij}:=\mathcal{I}^g_{ij}\cup \mathcal{I}^f_{ij}$. Differently from \eqref{eq:multi agent spec}, we highlight that after the decomposition each edge $(i,j)\in \mathcal{E}_{\bar{\psi}}$ can contain parametric tasks according to \eqref{eq:specific form of phi bar} and un-decomposed tasks from $\psi$ as per \eqref{eq:multi agent spec} that are directly inherited by $\bar{\psi}$. In order to differentiate between the formers and the latters, we introduce the sets $\bar{\mathcal{F}}$ and $\mathcal{F}$ such that $\phi^k_{ij}\in \bar{\mathcal{F}}$ if $\phi^k_{ij}$ is a task according to \eqref{eq:specific form of phi bar}, while  $\phi^k_{ij}\in \mathcal{F}$ if $\phi^k_{ij}$ is defined according to \eqref{eq:multi agent spec}. For any task $\phi^k_{ij}$ in \eqref{eq: always splitted specification}-\eqref{eq: eventually splitted specification} we define  the associated predicate function $h^k_{ij}(\vec{e}_{ij})$, predicate $\mu^k_{ij}(\vec{e}_{ij}):=\bigl\{\begin{smallmatrix*}[r]
\true& \; \text{ if} \quad h^k_{ij}(\vec{e}_{ij},\vec{\eta}_{ij})\geq 0 \\
\false& \; \text{ if} \quad h^k_{ij}(\vec{e}_{ij},\vec{\eta}_{ij})< 0 ,
\end{smallmatrix*} \label{eq:parameteric predicate}$ and super-level set $\mathcal{B}^k_{ij}=\{\vec{e}_{ij}|h^k_{ij}(\vec{e}_{ij})\geq0\}$, which is a parametric hyper-rectangle if $\phi^k_{ij}\in \mathcal{\bar{F}}$ as per \eqref{eq:parameteric superlevel set}. We now present the four types of conflicting conjunctions we consider. 
\begin{Fact}(Conflict of type 1)\label{conflict 1}
Consider two tasks $\phi^k_{ij}$ and $\phi^q_{ij}$ defined over the edge $(i,j)\in\mathcal{E}_{\bar{\psi}}$  and such that $k,q\in\mathcal{I}_{ij}^g$
If $[a^q,b^q]\cap[a^k,b^k]\neq\emptyset$ and $\mathcal{B}_{ij}^{k}\cap\mathcal{B}_{ij}^{q}=\emptyset$ then $\phi^k_{ij}\land\phi^q_{ij}$ is a conflicting conjunction.
\end{Fact}
\begin{proof}
We prove the fact by contradiction. Suppose there exists $\vec{x}(t)$ such that $\vec{x}(t)\models\phi^q_{ij}\land\phi^k_{ij}$. Then $\rho^{\phi^q_{ij}\land\phi^k_{ij}}(\vec{x},0)= \min \{\rho^{\phi^q_{ij}}(\vec{x},0),\rho^{\phi^k_{ij}}(\vec{x},0)\}>0 \Rightarrow \rho^{\phi^q_{ij}}
    (\vec{x},0)>0 \land  \rho^{\phi^k_{ij}}(\vec{x},0) > 0$. Recalling the definition of robust semantics for the $G$ operator we have $h^q_{ij}(\vec{e}_{ij}(t))>0 \; \forall t\in[a^q,b^q]\Rightarrow \vec{e}_{ij}(t)\in \mathcal{B}^q_{ij} \; \forall t\in[a^q,b^q]$ and $h^k_{ij}(\vec{e}_{ij}(t))>0 \; \forall t\in[a^k,b^k]\Rightarrow \vec{e}_{ij}(t)\in \mathcal{B}^k_{ij}\; \forall t\in[a^k,b^k]$. Since $[a^k,b^k]\cap [a^q,b^q] \neq \emptyset$  then there exists $\bar{t}\in [a^k,b^k]\cap [a^q,b^q]$ for which $\vec{e}_{ij}(\bar{t})\in \mathcal{B}^k_{ij} \land \vec{e}_{ij}(\bar{t})\in \mathcal{B}^q_{ij}$. We thus arrived at a contradiction since $\mathcal{B}^k_{ij}\cap \mathcal{B}^q_{ij}=\emptyset$.
\end{proof}
\begin{Fact}(Conflict of type 2)\label{conflict 2}
Consider two tasks $\phi^k_{ij}$ and $\phi^q_{ij}$ defined over the edge $(i,j)\in\mathcal{E}_{\bar{\psi}}$ and such that $k\in\mathcal{I}_{ij}^g$, $q\in\mathcal{I}_{ij}^f$.
If $[a^q,b^q]\subseteq[a^k,b^k]$ and $\mathcal{B}^q_{ij}\cap\mathcal{B}_{ij}^{k}=\emptyset$ then $\phi^k_{ij}\land\phi^q_{ij}$ is a conflicting conjunction.
\end{Fact}
\begin{proof}
 We prove the fact by contradiction. Suppose there exists $\vec{x}(t)$ such that $\vec{x}(t)\models\phi^q_{ij}\land\phi^k_{ij}$. Then  $\rho^{\phi^q_{ij}\land\phi^k_{ij}}(\vec{x},0)= \min \{\rho^{\phi^q_{ij}}(\vec{x},0),\rho^{\phi^k_{ij}}(\vec{x},0)\}>0 \Rightarrow \rho^{\phi^q_{ij}}
(\vec{x},0)>0 \land  \rho^{\phi^k_{ij}}(\vec{x},0)> 0$. Recalling the definition of robust semantics for the $F$ and $G$ operators we have $h^q_{ij}(\vec{e}_{ij}(t))>0 \; \forall t\in[a^q,b^q]\Rightarrow \vec{e}_{ij}(t)\in \mathcal{B}^q_{ij} \; \forall t\in[a^q,b^q]$ and there exist $\bar{t}\in[a^k,b^k]$ such that $h^k_{ij}(\vec{e}_{ij}(\bar{t}))>0 \Rightarrow \vec{e}_{ij}(\bar{t})\in \mathcal{B}^k_{ij}$. Since $[a^k,b^k]\subset [a^q,b^q]$ then $\bar{t}\in[a^q,b^q]$ and $\vec{e}_{ij}(\bar{t})\in \mathcal{B}^k_{ij} \land \vec{e}_{ij}(\bar{t})\in \mathcal{B}^q_{ij}$. We thus arrived at the contradiction as $\mathcal{B}^k_{ij}\cap \mathcal{B}^q_{ij}=\emptyset$.
\end{proof}
The next two facts define conflicting conjunctions over cycles of tasks in $\mathcal{G}_{\bar{\psi}}$. Indeed, conflicts may arise if the cycle closure relation \eqref{eq:edge sequence} can not be satisfied under a cycle of tasks in $\mathcal{G}_{\bar{\psi}}$. Namely, we deal with conjunction of tasks $\land_{(r,s)\in\epsilon(\vec{\omega})}\phi_{rs}$ where $\vec{\omega}$ is a cycle of length $l$ in $\mathcal{G}_{\bar{\psi}}$ and $\phi_{rs}$ is a task of type \eqref{eq: always splitted specification} or \eqref{eq: eventually splitted specification}. The case in which each $\phi_{rs}$ is a task of type \eqref{eq: splitted specification} that contains conjunctions, is a generalization of this simpler case as we show later. For clarity of presentation, we adopt the notation $\phi_{\omega[k,k+1]}$ to indicate a task $\phi_{rs}$ such that $(r,s)=(\omega[k],\omega[k+1])$ for $k\in1,\ldots l-1$. Likewise, we denote the corresponding time interval, the super-level set, predicate function and predicate of $\phi_{\omega[k,k+1]}$ as $[a,b]_{\omega[k,k+1]}$, $\mathcal{B}_{\omega[k,k+1]}$, $h_{\omega[k,k+1]}$ and $\mu_{\omega[k,k+1]}$ respectively. We thus have the notational equivalence $\land_{(r,s)\in\epsilon(\omega)}\phi_{rs} = \land_{k=1}^{l-1} \phi_{\omega[k,k+1]}$.
We now present the next two types of conflicting conjunctions.
\begin{Fact}(Conflict of type 3)\label{conflict 3}
Consider a cycle $\vec{\omega}$ of length $l$ in $\mathcal{G}_{\bar{\psi}}$ such that each edge $(\vec{\omega}[k],\vec{\omega}[k+1])$ corresponds to a unique task $\phi_{\omega[k,k+1]}=G_{[a,b]}\mu_{\omega[k,k+1]}$. 
If $\bigcap_{k=1}^{l-1}[a,b]_{\omega[k,k+1]}\neq\emptyset$ and 
\begin{equation}\label{eq:Minkowski condition for conflict}
\bigoplus_{k=1}^{p}\mathcal{B}_{\omega[k,k+1]}\cap \Bigl(- \bigoplus_{k=p+1}^{l-1}\mathcal{B}_{\omega[k,k+1]}\Bigr)= \emptyset
\end{equation}
for some $1\leq p\leq l-1$ then $\bigwedge_{k=1}^{l-1} \phi_{\omega[k,k+1]}$ is a conflicting conjunction.
\end{Fact}
\begin{proof}
We prove the theorem by contradiction. Assume that there exists a state trajectory $\vec{x}(t)$ such that $\vec{x}(t) \models \bigwedge_{k=1}^{l-1} \phi_{\omega[k,k+1]}$. Then we know from the robust satisfaction of such conjunction that $\rho^{\bigwedge_{k=1}^{l-1} \phi_{\omega[k,k+1]}}(\vec{x},0) = \min_{k=1,\ldots l-1} \{ \rho^{\phi_{\omega[k,k+1]}}(\vec{x},0) \} > 0 \Rightarrow \rho^{\phi_{\omega[k,k+1]}}(\vec{x},0)>0  \; \forall k=1,\ldots l-1$. Recalling the definition of robust semantics for the always formulas we can then write $h_{\omega[k,k+1]}(\vec{e}_{\omega[k,k+1]}(t))>0,  
 \forall t\in[a,b]_{\omega[k,k+1]},\forall k=1,\ldots l-1$ and thus $\vec{e}_{\omega[k,k+1]}(t)\in \mathcal{B}_{\omega[k,k+1]},  \forall t\in[a,b]_{\omega[k,k+1]},\forall k=1,\ldots l-1$. We now recall from \eqref{eq:edge sequence} that for a cycle of edges, we have $\sum_{k=1}^{l-1}\vec{e}_{\omega[k,k+1]}(t) = 0 \Rightarrow \sum_{k=1}^{p}\vec{e}_{\omega[k,k+1]}(t) = - \sum_{k=p+1}^{l-1}\vec{e}_{\omega[k,k+1]}(t)$ for any $1\leq p\leq l-2$. Since we know that $\bigcup_{k=1}^{l-1} [a,b]_{\omega[k,k+1]} \neq \emptyset$ then there exist a time instant $\bar{t}\in \bigcup_{k=1}^{l-1} [a,b]_{\omega[k,k+1]}$ such that the three relations $\sum_{k=1}^{p}\vec{e}_{\omega[k,k+1]}(\bar{t})\in \bigoplus_{k=1}^{p}\mathcal{B}_{\omega[k,k+1]}$, $-\sum_{k=p+1}^{l-1}\vec{e}_{\omega[k,k+1]}(\bar{t}) \in \bigl(-\bigoplus_{k=p+1}^{l-1}\mathcal{B}_{\omega[k,k+1]}\bigr)$ and $\sum_{k=1}^{p}\vec{e}_{\omega[k,k+1]}(\bar{t}) = - \sum_{k=p+1}^{l-1}\vec{e}_{\omega[k,k+1]}(\bar{t})$ must hold. We thus arrived at a contradiction as the three aforementioned conditions can not hold jointly if $\bigoplus_{k=1}^{p}\mathcal{B}_{\omega[k,k+1]}\cap \Bigl(- \bigoplus_{k=p+1}^{l-1}\mathcal{B}_{\omega[k,k+1]}\Bigr)= \emptyset$. Since the argument is independent of the index $p$ chosen, the proof is valid for any chosen $1\leq p\leq l-2$.
\end{proof}
\begin{Fact}(Conflict of type 4)\label{conflict 4}
Consider a cycle $\vec{\omega}$ of length $l$ in $\mathcal{G}_{\psi}$ such that for each edge $(\vec{\omega}[k],\vec{\omega}[k+1])$ there corresponds a unique task $\phi_{\omega[k,k+1]}$, for which we have $\phi_{\omega[k,k+1]}=F_{[a,b]}\mu_{\omega[k,k+1]} \forall k=1,\ldots q$ and $\phi_{\omega[k,k+1]}=G_{[a,b]}\mu_{\omega[k,k+1]}\forall k=q+1,\ldots l-1$ for some $1\leq q\leq l-1$. If $\bigoplus_{k=1}^{p}\mathcal{B}_{\omega[k,k+1]}\cap \Bigl(-\bigoplus_{k=p+1}^{l-1}\mathcal{B}_{\omega[k,k+1]}\Bigr)\neq \emptyset$ for some $1\leq p\leq l-1$ and if either 1) $q=1$, $\bigcap_{k=q+1}^{l-1}[a,b]_{\omega[k,k+1]}\neq\emptyset$, $[a,b]_{\omega[1,2]}\subseteq \bigcap_{k=q+1}^{l-1}[a,b]_{\omega[k,k+1]}$ or 2) $q\geq1$, $\bigcap_{k=q+1}^{l-1}[a,b]_{\omega[k,k+1]}\neq\emptyset$, $[a,b]_{\omega[k,k+1]}=[\bar{t},\bar{t}] \forall k={1,\ldots q}$, $\bar{t}\in \bigcap_{k=q+1}^{l-1}[a,b]_{\omega[k,k+1]}$; then  $\bigwedge_{k=1}^{l-1} \phi_{\omega[k,k+1]}$ is a conflicting conjunction.
\end{Fact}
\begin{proof}
The proof is similar to the proof of Fact \ref{conflict 3} and it is not reported due to space limitation.
\end{proof}
\subsection{Resolving conflicting conjunctions}
During the process of rewriting $\psi$ into $\bar{\psi}$, it is relevant to avoid the insurgence of conflicts 1-4. With this aim, we provide the Lemmas \ref{lemma type 1}-\ref{lemma type 4} which define additional convex constraints on the parameters of the parametric tasks in $\bar{\mathcal{F}}$. We then include such constraints in \eqref{eq:convex optimization problem} to avoid conflicts of type 1-4.
\begin{lemma}\label{lemma type 1}
Let $(i,j) \in \mathcal{E}_{\bar{\psi}}$ with task $\phi_{ij}=\phi_{ij}^s \land \phi_{ij}^r$ such that $s,r\in\mathcal{I}_{ij}^g$ and $[a^r,b^r]\cap[a^s,b^s]\neq\emptyset$.  Then, enforcing the convex constraints
\begin{subequations}\label{eq:type 1 constraint}
\begin{align}
\begin{split}\label{eq: first case}
&\mathcal{B}^s_{ij}\subseteq\mathcal{B}^r_{ij} \; \text{if} \quad (b^r-a^r)\geq (b^s- a^s)\land \phi^r_{ij},\phi^s_{ij}\in\bar{\mathcal{F}}
\end{split}\\
\begin{split}\label{eq: second case}
&\mathcal{B}^r_{ij}\subseteq\mathcal{B}^s_{ij} \; \text{if} \quad (b^r-a^r)< (b^s- a^s)\land \phi^r_{ij},\phi^s_{ij}\in\bar{\mathcal{F}}
\end{split}\\
\begin{split}
&\mathcal{B}^s_{ij}\subseteq\mathcal{B}^r_{ij} \; \text{if} \quad \phi^r_{ij}\in\mathcal{F},\phi^s_{ij}\in\bar{\mathcal{F}}
\end{split}
\end{align}
\end{subequations}
where $\mathcal{F}$ and $\bar{\mathcal{F}}$ are the sets of non-parametric and parametric tasks respectively, ensures that $\phi_{ij}^s \land \phi_{ij}^r$ is not a conflicting conjunction of type 1.
\end{lemma}
\begin{proof}
Let $[a^r,b^r]\cap[a^s,b^s]\neq\emptyset$, then  \eqref{eq:type 1 constraint} contradicts the conditions of conflicting conjunction of type 1 according to Fact \ref{conflict 1}.
\end{proof}
Constraints \eqref{eq: first case}-\eqref{eq: second case} specify that when two parametric formulas could be in conflict according to Fact \ref{conflict 1}, then we decide to include the super-level set of tasks with shorter time intervals into the ones with longer time intervals. 
\begin{lemma} \label{lemma type 2}
Let $(i,j) \in \mathcal{E}_{\bar{\psi}}$ with task $\phi_{ij}=\phi_{ij}^s \land \phi_{ij}^r$ such that $s\in\mathcal{I}_{ij}^g$,  $r\in\mathcal{I}_{ij}^f$ and $[a^r,b^r]\subseteq[a^s,b^s]$. Then, enforcing the convex constraints
\begin{subequations}\label{eq:type 2 constraints}
\begin{align}
\begin{split}\label{eq:type 2 constraint case 1}
\mathcal{B}^r_{ij}\subseteq\mathcal{B}^s_{ij} \quad \text{if} \quad \phi^s_{ij}\in\mathcal{F},\phi^r_{ij}\in\mathcal{\bar{F}}
\end{split}\\
\begin{split}\label{eq:type 2 constraint case 2}
\mathcal{B}^s_{ij}\subseteq\mathcal{B}^r_{ij} \quad \text{if} \quad \phi^s_{ij}\in\mathcal{\bar{F}},\phi^r_{ij}\in\mathcal{F}
\end{split}\\
\begin{split}\label{eq:type 2 constraint case 2}
\mathcal{B}^r_{ij}\subseteq\mathcal{B}^s_{ij} \quad \text{if} \quad \phi^s_{ij}\in\mathcal{\bar{F}},\phi^r_{ij}\in\mathcal{\bar{F}}
\end{split}
\end{align}
\end{subequations}
where $\mathcal{F}$ and $\bar{\mathcal{F}}$ are the sets of non-parametric and parametric tasks respectively, ensure that $\phi_{ij}^s \land \phi_{ij}^r$ is not a conflicting conjunction of type 2.
\end{lemma}
\begin{proof}
Let $[a^r,b^r]\subseteq[a^s,b^s]$, then \eqref{eq:type 2 constraints} contradicts the conditions of conflicting conjunction of type 2 according to Fact \ref{conflict 2}.
\end{proof}
When more than two tasks are considered in conjunction, then Lemma \ref{lemma type 1}-\ref{lemma type 2} can be applied to every possible couple of tasks in the conjunction.
\begin{lemma} \label{lemma type 3}
Let $\vec{\omega}$ be a cycle of length $l$ in $\mathcal{G}_{\bar{\psi}}$ such that for each $(\vec{\omega}[k],\vec{\omega}[k+1])$ there corresponds a unique task $\phi_{\omega[k,k+1]}=G_{[a,b]}\mu_{\omega[k,k+1]}$. If $\bigcap_{k=1}^{l-1}[a,b]_{\omega[k,k+1]}\neq\emptyset$ then imposing the constraints 
\begin{equation}\label{eq:type 3 constraint}
\bigoplus_{k =1}^p\mathcal{B}_{\omega[k,k+1]}\subseteq \Bigl(-\bigoplus_{k=p+1}^{l-1}\mathcal{B}_{\omega[k,k+1]}\Bigr)
\end{equation}
for some $1\leq p\leq l-1$, ensures that $\bigwedge_{k=1}^{l-1}\phi_{\omega[k,k+1]}$ is not a conflicting conjunction of type 3.
\end{lemma}
\begin{proof}
Relation \eqref{eq:type 3 constraint} is a special case of the set inclusion in \eqref{eq:Minkowski condition for conflict}. Since we considered $\bigcap_{k=1}^{l-1}[a,b]_{\omega[k,k+1]}\neq\emptyset$, then \eqref{eq:type 3 constraint} contradicts \eqref{eq:Minkowski condition for conflict} as per Fact \ref{conflict 3}.
\end{proof}
\begin{lemma} \label{lemma type 4}
Let $\vec{\omega}$ be a cycle of length $l$ in $\mathcal{G}_{\psi}$ such that for each edge $(\vec{\omega}[k],\vec{\omega}[k+1])$ there corresponds a unique task $\phi_{\omega[k,k+1]}$. Arbitrarily assume $\phi_{\omega[k,k+1]}=F_{[a,b]}\mu_{\omega[k,k+1]} \forall k=1,\ldots q$ and $\phi_{\omega[k,k+1]}=G_{[a,b]}\mu_{\omega[k,k+1]}\forall k=q+1,\ldots l-1$ for some $1\leq q\leq l-1$. If either 1) $q=1$, $\bigcap_{k=q+1}^{l-1}[a,b]_{\omega[k,k+1]}\neq\emptyset$, $[a,b]_{\omega[1,2]}\subseteq \bigcap_{k=q+1}^{l-1}[a,b]_{\omega[k,k+1]}$, or 2) $q\geq1$, $\bigcap_{k=q+1}^{l-1}[a,b]_{\omega[k,k+1]}\neq\emptyset$, $[a,b]_{\omega[k,k+1]}=[\bar{t},\bar{t}] \forall k={1,\ldots q}$, $\bar{t}\in \bigcap_{k=q+1}^{l-1}[a,b]_{\omega[k,k+1]}$;
then imposing \eqref{eq:type 3 constraint} ensures that $\bigwedge_{k=1}^{l-1}\phi_{\omega[k,k+1]}$ is not a conflicting conjunction of type 4. 
\end{lemma}
\begin{proof}
The proof is similar to the proof of Lemma \ref{lemma type 3} and it is not reported due to space limitations.
\end{proof}
\begin{remark}
In the statement of Lemmas \ref{lemma type 3}-\ref{lemma type 4} we considered that no edge of a cycle $\vec{\omega}$ contains conjunctions. Nevertheless, in the case $\phi_{\omega[k,k+1]}$ contains conjunctions as $\phi_{\omega[k,k+1]}=\bigwedge_{s\in \mathcal{I}^g_{\omega[k,k+1]}}\phi^s_{\omega[k,k+1]} \land \bigwedge_{s\in \mathcal{I}^f_{\omega[k,k+1]}}\phi^s_{\omega[k,k+1]}$ according to \eqref{eq: splitted specification}, it is possible to define the combination set $C = \prod_{k=1}^{l-1}\mathcal{I}_{\omega[k,k+1]}$ where $\mathcal{I}_{\omega[k,k+1]}$ contains all the tasks indices for edge $(\vec{\omega}[k],\vec{\omega}[k+1])$ as per \eqref{eq: splitted specification}. Then, for each combination of indices  $c\in \mathcal{C}$ we can check if the conjunction $\land_{k=1}^{l-1}\phi^{c[k]}_{\omega[k,k+1]}$ satisfies the conditions for conflicting conjunctions of type 3-4 and introduce new constraints of type \eqref{eq:type 3 constraint} accordingly. 
\end{remark}
\begin{remark}
The computation of the Minkowski sum in \eqref{eq:type 3 constraint} can possibly become a complex computation as not all the tasks defined over a cycle $\vec{\omega}$ are parametric as per \eqref{eq:infty predicate function}. Nevertheless, we can under-approximate the sets $\mathcal{B}_{\omega[k,k+1]}$ such that $\phi_{\omega[k,k+1]}\in\mathcal{F}$ by a hyper-rectangle that can be computed offline when $\psi$ is defined. \cite{ziegler2012lectures}.
\end{remark}
We conclude this section with the final result of our work.
 \begin{theorem}\label{the last fermat theorem}
     Let the conditions of Theorem \ref{convex optimization theorem} be satisfied. If a solution to the convex optimization problem \eqref{eq:convex optimization problem} exists after the inclusion of constraints \eqref{eq:type 1 constraint}, \eqref{eq:type 2 constraints}, \eqref{eq:type 3 constraint}, then the resulting global formula $\bar{\psi}$ defined as in \eqref{eq:specific form of phi bar} from Thm. \ref{convex optimization theorem} does not contain conflicting conjunctions as per Fact \ref{conflict 1}-\ref{conflict 4}.
 \end{theorem}
 \begin{proof}
Follows directly from Lemmas \ref{lemma type 1}-\ref{lemma type 4}.
 \end{proof}
Thus, Theorem \ref{the last fermat theorem} resolves the feasibility issues that were not considered in Theorem \ref{convex optimization theorem} so that $\bar{\psi}$ does not suffer from conflicting conjunctions according to Fact \ref{conflict 1}-\ref{conflict 4}.

\section{SIMULATIONS}\label{simulations}
\begin{figure*}[h]
\begin{subfigure}[t]{.33\textwidth}
\centering
  \includegraphics[width=1\linewidth]{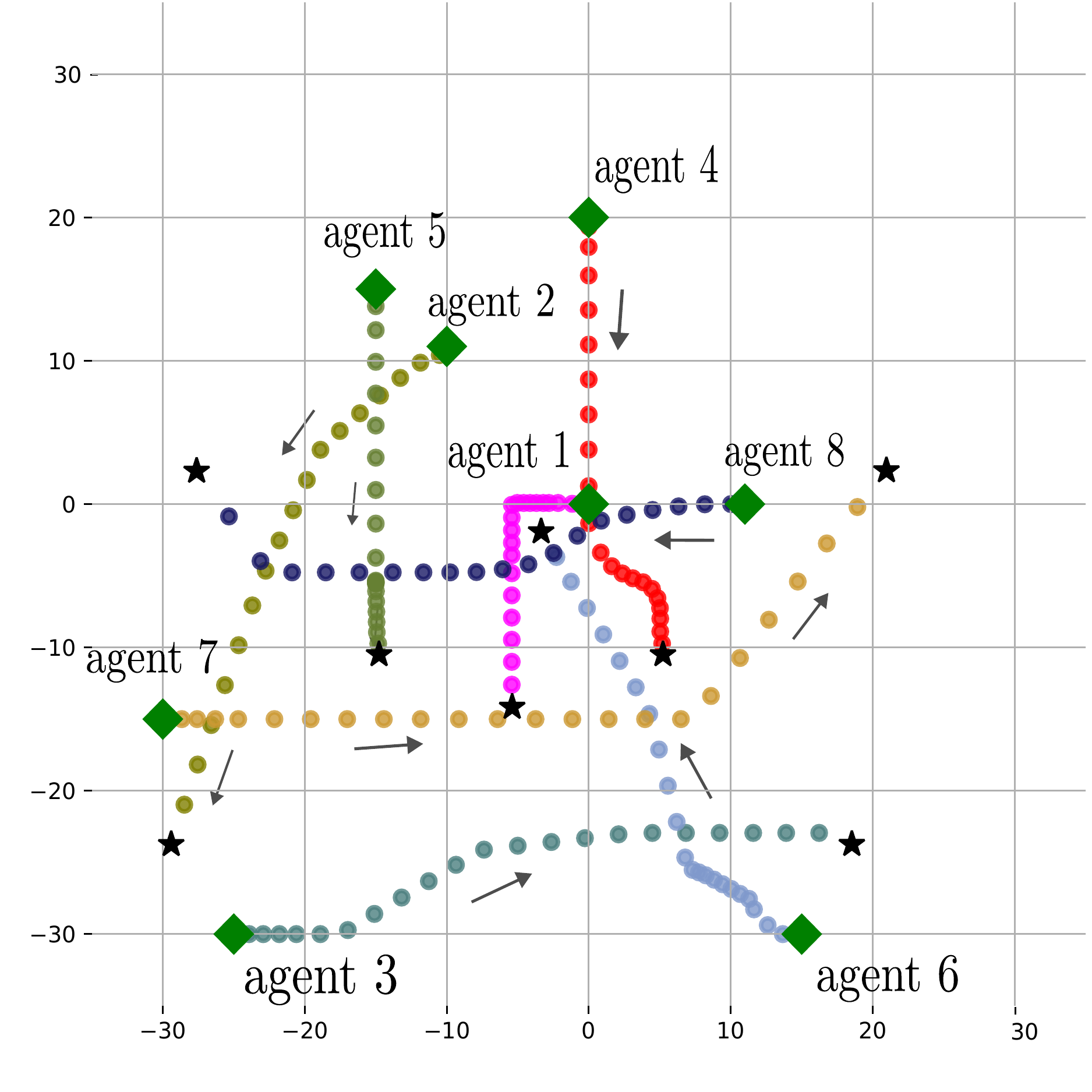}
  \caption{}
  \label{fig:t25}
\end{subfigure}\hspace{1em}%
\begin{subfigure}[t]{.33\textwidth}
\centering
\includegraphics[width=\linewidth,trim = {0cm 0cm 0cm 1cm}]{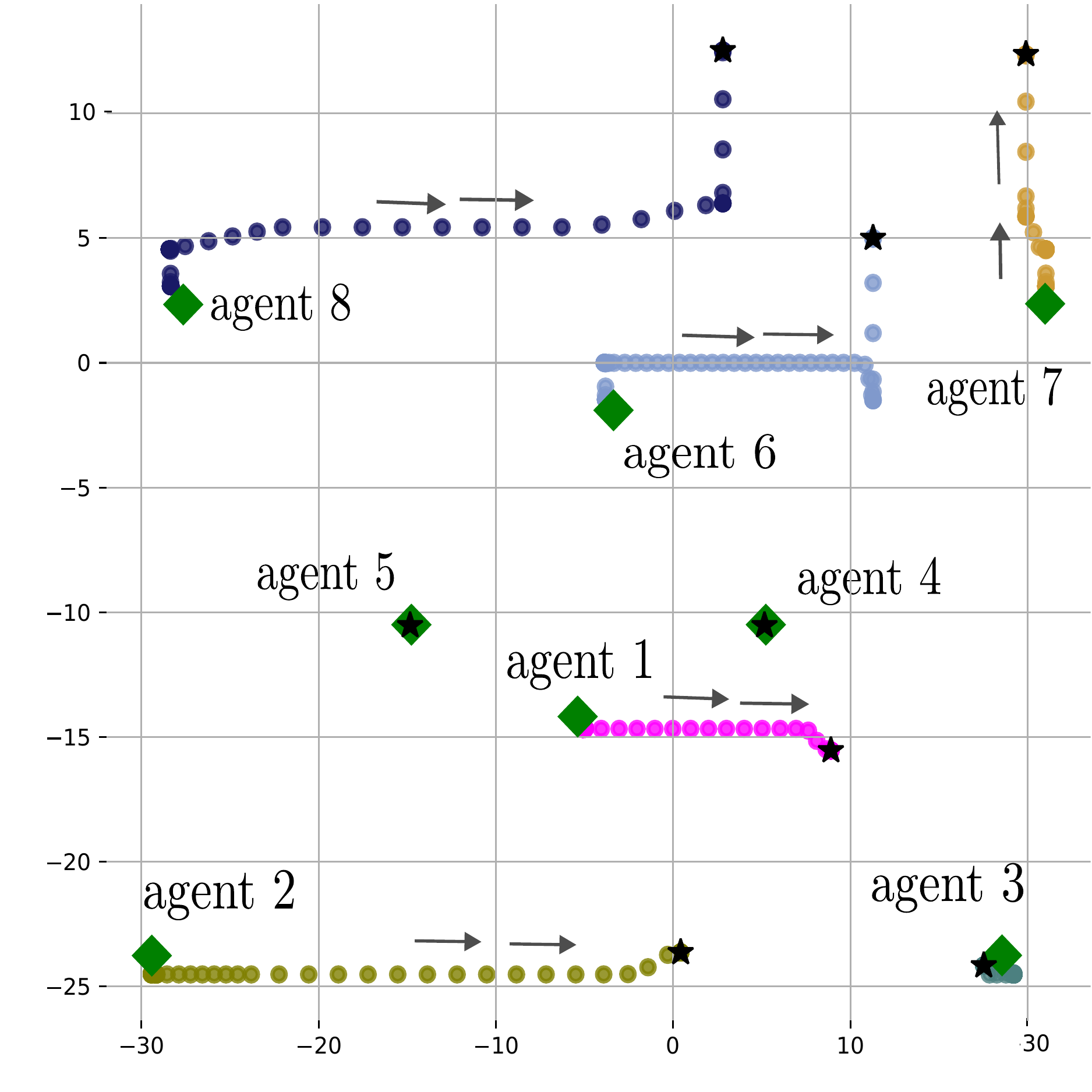}
  \caption{}
  \label{fig:t35}
\end{subfigure}\hspace{1em}%
\begin{subfigure}[t]{.33\textwidth}
\centering
    \includegraphics[width=1\linewidth]{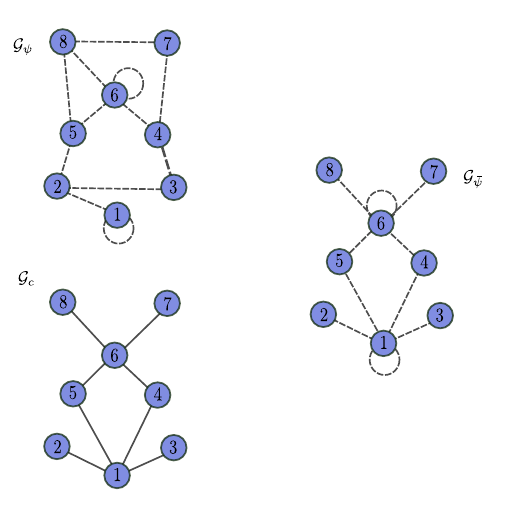}
    \caption{}
    \label{fig: sim2 decomposition}
\end{subfigure} 
\caption{Trajectory evolution of the agents from time $t=0s$ to $t=15s$ (a) and from time $t=15s$ to $t=28s$ (b); Short solid arrows represent the direction of movement of the agents; green lozenges represent the starting positions of the agents, while  black stars represent the final positions.  Communication graph $\mathcal{G}_c$, initial task graph $\mathcal{G}_{\psi}$ and final task graph $\mathcal{G}_{\bar{\psi}}$ (c). Edges $(2,5), (3,2), (3,4), (7,4), (8,5)$, and $(7,8)$ in $\mathcal{G}_\psi$ are decomposed over the edges of $\mathcal{G}_c$ to obtain $\mathcal{G}_{\bar{\psi}}$.}
\vspace{-0.3cm}
\end{figure*}
We showcase our decomposition approach for a formation problem of 8 agents such that $\mathcal{V}:=\{1,\ldots8\}$. For this task, we assume the agents to be governed by single integrator dynamics $\dot{\vec{x}}_i=\vec{u}_i$ where $\vec{x}_i=[x_i,y_i]^T\in\mathbb{R}^2$ represents the position of the agents in a 2-dimensional space. We consider a time-varying formation task divided into two parts. First, a star formation needs to be achieved  given by the tasks: $\phi_{85}=G_{[10,15]}(||\vec{e}_{85}-[-15,15]^T||\leq3)$,
$\phi_{52}=G_{[10,15]}(||\vec{e}_{25}-[-15,-15]^T||\leq3)$, $\phi_{34}=G_{[10,15]}(||\vec{e}_{34}-[15,-15]^T||\leq3)$, $\phi_{74}=G_{[10,15]}(||\vec{e}_{74}-[15,15]^T||\leq3)$, $\phi_{46}=G_{[10,15]}(||\vec{e}_{46}-[10,-10]^T||\leq2)$, $\phi_{56}=G_{[10,15]}(||\vec{e}_{56}-[-10,-10]^T||\leq2)$. Second, the team of agents $\mathcal{V}_1=\{6,7,8\}$ and $\mathcal{V}_2=\{1,2,3\}$ get detached by the initial star formation such that agents $1$ and $6$  move independently toward the region $x>5,y>5$ and $x>5,y<-5$ respectively while team $\mathcal{V}_1$ and $\mathcal{V}_2$ achieve a triangle formation from $t=25s$ to $t=28s$. This second part of the formation is given by specifications $\phi_{1}=F_{[25,28]}([-1,0](\vec{x}_{1} - [5,0]^T)\leq0 )\wedge ([0,-1](\vec{x}_{1} - [0,-5]^T)\leq0$, $\phi_{6}=F_{[25,28]}([-1,0](\vec{x}_{6} - [5,0]^T)\leq0) \wedge ([0,-1](\vec{x}_{6} - [0,5]^T)\leq0$, $\phi_{32}=G_{[25,28]}(||\vec{e}_{32}-[16,0]^T||\leq2\sqrt{2})$, $\phi_{87}=G_{[25,28]}(||\vec{e}_{87}-[16,0]^T||\leq2\sqrt{2})$, $\phi_{21}=G_{[25,28]}(||\vec{e}_{21}-[-8,-8]^T||\leq2\sqrt{2})$,
$\phi_{68}=G_{[25,28]}(||\vec{e}_{68}-[-8,8]^T||\leq2\sqrt{2})$. The global task $\psi$ is the conjunction of the given tasks. 
We use the open-source optimization library \textit{CasADi} which leverages an interior-point method algorithm to solve the convex program in \eqref{eq:convex optimization problem} and find $\bar{\psi}$. The algorithm converged to an optimal solution in $0.019s$ running on an Intel-Core i7-1265U × 12 with 32 GB of RAM. In order to satisfy $\bar{\psi}$, we apply a fully distributed control barrier function-based controller similar to the approach proposed previously in \cite{LarsControl2} so that each agent $i\in\mathcal{V}$ needs only state information of its 1-hop neighbours in $\mathcal{N}_c(i)$ to compute its local control $\vec{u}_i$. The graphs $\mathcal{G}_{\psi}$, $\mathcal{G}_{c}$ and $\mathcal{G}_{\bar{\psi}}$ for the simulation are given in Fig. \ref{fig: sim2 decomposition}. The tasks $\phi_{25}$, $\phi_{32}$, $\phi_{34}$, $\phi_{74}$, $\phi_{85}$ and $\phi_{78}$ are decomposed according to Table \ref{tab:decomposed formulas parameters}. The agents' trajectories are simulated from $t=0$ to $t=28$ and shown in Fig. \ref{fig:t25}-\ref{fig:t35}.
\begin{table}[h]
\begin{tabularx}{\linewidth}{ccccc}
\toprule\toprule
     & \multicolumn{2}{c}{$\phi^{\pi_{5}^2}\quad$$\vec{\pi}_{5}^2=[5,1,2]$ }     & \multicolumn{2}{c}{$\phi^{\pi_{2}^3}\quad$$\vec{\pi}_{2}^3=[2,1,3]$}                                                 \\
\cmidrule(rl){2-3} \cmidrule(rl){4-5}
            & $\bar{\phi}^{\pi_{5}^2}_{51}$ & $\bar{\phi}^{\pi_{5}^2}_{12}$ & $\bar{\phi}^{\pi_{2}^3}_{21}$ & $\bar{\phi}^{\pi_{2}^3}_{13}$                                           \\
$\vec{p}_{rs}^{\pi_i^j}$   & $[9.99 , \text{-}4.34]$  & $[24.99 ,\text{-}10.65]$ & $[7.99,7.99]$     & $[7.99, \text{-}7.99]$                                             \\
$\vec{\nu}_{rs}^{\pi_i^j}$ & $[0.47,0.47]$     & $[1.65,1.62]$     & $[0.95,0.95]$     & $[1.25,1.26]$  
\\
\end{tabularx}
\begin{tabularx}{\linewidth}{lcccc}
\midrule
     & \multicolumn{2}{c}{$\phi^{\pi_{4}^3}\quad$$\vec{\pi}_{4}^3=[4,1,3]$}       & \multicolumn{2}{c}{$\phi^{\pi_{4}^7}\quad$$\vec{\pi}_{4}^7=[4,6,7]$}                                                 \\
\cmidrule(rl){2-3} \cmidrule(rl){4-5}
            & $\bar{\phi}^{\pi_{4}^3}_{41}$ & $\bar{\phi}^{\pi_{4}^3}_{13}$ & $\bar{\phi}^{\pi_{4}^7}_{46}$ & $\bar{\phi}^{\pi_{4}^7}_{67}$  
            \\
$\vec{p}_{rs}^{\pi_i^j}$   & $[\text{-}9.99,\text{-}4.35]$   & $[24.99, \text{-}10.65]$ & $[\text{-}9.99,9.99]$    & $[25.00,5.00]$                                              \\
$\vec{\nu}_{rs}^{\pi_i^j}$ & $[0.47,0.47]$     &       $[1.62, 1.62]$            & $[1.89,1.89]$     & $[0.91, 0.91]$\\
\end{tabularx}
\begin{tabularx}{\linewidth}{lcccc}
\midrule
     & \multicolumn{2}{c}{$\phi^{\pi_{5}^8}\quad$$\vec{\pi}_{5}^8=[5,6,8]$}       & \multicolumn{2}{c}{$\phi^{\pi_{8}^7}\quad$$\vec{\pi}_{8}^7=[8,6,7]$}                                                 \\
\cmidrule(rl){2-3} \cmidrule(rl){4-5}
            & $\bar{\phi}^{\pi_{5}^8}_{56}$ & $\bar{\phi}^{\pi_{5}^8}_{68}$ & $\bar{\phi}^{\pi_{8}^7}_{67}$ & $\bar{\phi}^{\pi_{8}^7}_{86}$                                           \\
      $\vec{p}_{rs}^{\pi_i^j}$      & $[10.00, 10.00]$  & $[\text{-}25.00 ,5.00]$  & $[7.99, \text{-}7.99]$   & $[7.99,  7.99]$                                             \\
    $\vec{\nu}_{rs}^{\pi_i^j}$        & $[1.89, 1.89]$    & $[0.91, 0.91]$   & $[0.95 0.95]$     & $[1.26,1.26]$  \\
            \midrule
\end{tabularx}
\caption{Decomposed tasks $\phi_{25}$, $\phi_{32}$, $\phi_{34}$, $\phi_{74}$, $\phi_{85}$ and $\phi_{78}$ with respective parameters.}
\label{tab:decomposed formulas parameters}
\vspace{-0.5cm}
\end{table}
\section{CONCLUSIONS}\label{conclusions}
We presented a framework STL task decomposition based on the communication graph of multi-agent systems. We showed that the newly defined tasks imply the original ones when a valid decomposition is obtained. We additionally gave a set of Facts describing four possible conflicting conjunctions related to the applied STL fragment. We then provided a set of convex constraints, under which such conflicts can be avoided during the decomposition. In future work, we will investigate task decomposition under time-varying communication topology and more complex parameterization for task decomposition.


\bibliographystyle{ieeetr} 
\bibliography{references} 

\end{document}